\documentclass[12pt,letterpaper]{article}

\usepackage{amsmath}                    
\usepackage{amssymb}                    
\usepackage{amsthm}                     
\usepackage{amsfonts}                   
\usepackage{graphicx} 
\usepackage[dvips]{color}

\theoremstyle{definition}
\newtheorem{dfn}{Definition}
\theoremstyle{plain}
\newtheorem{thm}[dfn]{Theorem}

\newtheorem{lemma}[dfn]{Lemma}

\newtheorem{fact}[dfn]{Fact}

\newsavebox{\fmbox}
\newenvironment{fmpage}[1]
{\begin{lrbox}{\fmbox}\begin{minipage}{#1}}     
{\end{minipage}\end{lrbox}\fbox{\usebox{\fmbox}}}

\begin{document}

\title{Efficient Generation $\epsilon$-close to  $G(n,p)$ and Generalizations}

\author{Antonio Blanca\thanks{ U.C. Berkeley,  Email:ablanca@eecs.berkeley.edu,
Supported in part by NSF-CCF-TF-0830683.}
\and Milena Mihail\thanks{Georgia Tech,  Email:mihail@cc.gatech.edu
Supported in part by NSF-CCF-TF-0830683.}}

\date{}

\maketitle
\thispagestyle{empty} 

\begin{abstract}

We give an efficient algorithm to generate a graph from a distribution $\epsilon$-close to $G(n,p)$, in the sense of total variation distance. In particular, if $p$ is represented with $O(\log n)$-bit accuracy, then, with high probability, the running time is linear in the expected number of edges of the output graph (up to poly-logarithmic factors). All our running times include the complexity of the arithmetic involved in the corresponding algorithms. Previous standard methods for exact $G(n,p)$ sampling (see e.g.~\cite{bb}) achieve similar running times, however, under the assumption that performing real number arithmetic with arbitrary accuracy takes constant time. We note that the actual accuracy required by these methods is $O(n)$-bit  per step, which results in quadratic running times. We also note that compromising on the $O(n)$-bit accuracy requirement causes arbitrary large biases in the sampling.

The main idea of our $G(n,p)$ generation algorithm is  a  Metropolis Markov chain to sample $\epsilon$-close from the binomial distribution. This is a new method for sampling from the binomial distribution: it is of separate interest and may find other useful applications. Our analysis accounts for all necessary bit-accuracy and arithmetic. Our running times are  comparable to known methods for exact binomial sampling (e.g. surveyed in \cite{ks}), however, only when the latter do not account for bit-accuracy and assume arbitrary bit arithmetic in constant time. Dropping these assumptions affects their running times and/or causes  bias which has never been quantified. In this sense, our work can be viewed as a rigorous quantification of the tradeoff between accuracy and running time, when all computational aspects are taken into account. 

We further obtain efficient generation algorithms for random graphs with given arbitrary degree distributions, Inhomogeneous Random Graphs when the kernel function is the inner product, and Stochastic Kronecker Graphs. Efficient generation of these random graph models is essential for modeling large scale complex networks. To the best our knowledge, our work can be viewed as the first effort to simulate efficient generation of graphs from classical random graph models, while taking into account  implementational considerations as fundamental computational aspects, and quantifying the tradeoff between accuracy and running time in a way that can be useful in practice.

\end{abstract}

\null
\vfill

\pagebreak

\section{Introduction}
\label{intro}

Let $n\! \geq\! 1$ be an integer and let ${\cal G}_n$ be the set of 
$2^{{n \choose 2}}$ undirected simple graphs on $n$ vertices. 
For $0\! \leq p\! \leq 1$, we typically think of $G(n,p)$ as a graph in ${\cal G}_n$,
where each edge $\{u,v\}$ is present with probability $p$ and absent with probability $(1\! - \! p)$,
independently from all other edges. The straightforward way to generate such a graph
involves ${n \choose 2}$ independent experiments.
Thus the running time is $\Omega (n^2)$.
In practice, we need an additional $\Omega ( \log n)$ to represent $n$ distinct vertices 
and an additional $\Omega (\log p^{-1} )$ to simulate sampling with probability $p$ 
resulting in $\Omega \left( n^2 (      \log n +     \log p^{-1}      ) \right)$ total running time. 

If $m$ is the number of edges of the output graph, 
it is highly desirable to aim for  $O(m)$ running times, 
especially when ${\rm E}[ m ]  << n^2$ and $n$ is very large. 
This is true, for example, in the case of complex networks.
In these cases, the number of vertices $n$ is known to scale massively, 
while the corresponding  graphs remain relatively sparse~\cite{ab,bjr,cf,d,vdh}
(for example, due to natural underlying resource considerations,
as has been discussed extensively in the literature).
Taking into account the implementational issues mentioned in the previous paragraph
and the resource considerations,
we should aim for  $O \left( m (      \log n +     \log p^{-1}      ) \right)$ running times. 

Let $w_1 ,  \ldots , w_n$ be non-negative weights on $n$ vertices. Thus we may refer to the $w_u$'s as an $n$-dimensional vector $\vec{w}$. 
Let $G(n, \vec{w})$ be a graph in ${\cal G}_n$,
where each edge $\{ u,v\}$ is present with probability $\min \{w_u w_v , 1 \}$,  
independently from all other edges. Notice that $G(n,p)$ is a special case 
of $G(n, \vec{w})$, where $w_u \! = \! \sqrt{p}$ for all $u$. 
Moreover, $G(n, \vec{w})$ subsumes the case of random graphs with given expected 
degrees~\cite{cl2,cl,cl1,clv,mp,mh}. 

Let $d\! \geq \! 1$ be an integer and let $\vec{w_1} ,  \ldots , \vec{w_n}$ be vectors in $d$-dimensional
real space with non-negative coordinates: $w_{uk} \geq 0$,  
$\forall 1\! \leq \! u \! \leq \! n $ , $\forall 1\! \leq\!  k \! \leq \! d $.
We may thus refer to the $\vec{w_u}$'s using an $n\times d$ matrix $W$. 
Let $G(n,W)$ be a graph in ${\cal G}_n$,
where each edge $\{ u,v \}$ is present with probability $\min \{  \langle w_u , w_v  \rangle, 1  \}$,  
independently from all other edges, and where $\langle w_u , w_v  \rangle$ is the usual inner product 
$\langle w_u , w_v  \rangle  \! = \! \sum_{k=1}^d w_{uk} w_{vk}$. 
Notice that $G(n, \vec{w})$ is a special case of $G(n,W)$, where $d\! = \! 1$.
Moreover, $G(n, W)$ is a special case of Inhomogeneous Random Graphs~\cite{bjr,bjr1},
where the kernel function is the inner product. Random inner product graphs have been also studied in~\cite{ys,y}.
In the context of complex networks, the interpretation of vertices represented as $d$-dimensional vectors is natural.
Real datasets are categorical. Therefore, each dimension represents a distinct attribute, 
and data points connect with probabilities related to their similarity according to such attributes.
 
Let $k,d\! \geq \! 1$ be integers, and let ${\cal P}$ be a $d \times d$ initiator matrix. Define recursively the matrix $K_k \!=\! {\cal P} \otimes K_{k-1},$ where $K_0 = I$ and $\otimes$ is the Kronecker product of matrices \cite{leckfz}. For $n \!=\! d^k,$ let $G(n,{\cal P})$ be a graph in ${\cal G}_n$, where each edge $\{ u,v \}$ is present with probability $\min \{K_k(u,v),1\},$ independently from all other edges. $G(n,{\cal P})$ is known in the literature as Stochastic Kronecker Graphs \cite{ik,leckf,leckfz,lef,mx}.

In the context of complex networks, the model of random graphs with given expected degrees $G(n,\vec{w})$, 
the model of Inhomogeneous Random Graphs where the kernel function is the inner product $G(n,W)$,
and the Stochastic Kronecker model $G(n,{\cal P})$ have been used to produce 
synthetic graphs that capture important structural properties of real networks. 
Thus, efficient algorithms to generate such random graphs are important
both in theory and in practice.

Let $\pi$ be a probability distribution on ${\cal G}_n$. Random graph models, such as $G(n,p)$, $G(n,\vec{w}),$ $G(n,W),$ and $G(n,{\cal P})$ defined above, 
are equivalent to such distributions $\pi$ over ${\cal G}_n$. Moreover, properties of random graphs in such models are typically expressed 
as holding with high probability. This means that, for some constant $c\! > \!0$, 
the subset ${\cal G}_{\rm BAD}$ of graphs in ${\cal G}_n$ for which the property does not hold has 
$\pi (  {\cal G}_{\rm BAD} ) \leq n^{-c}$~\cite{b,d}. 
Such quantification is of fundamental predictive value, both in theory and applications.

Let $\pi^\prime$ be a probability distribution on ${\cal G}_n$ that is $\epsilon$-close to $\pi$, 
in the usual sense of variation distance: 
$\max_{{\cal H} \subset {\cal G}_n } |  \pi^\prime ( {\cal H}) - \pi ({\cal H}   ) | \leq \epsilon$. 
In the context of large scale complex networks it is typical to check experimentally 
for various desired (or not desired) properties of graphs generated according to a target distribution $\pi$. 
If generation was done according to $\pi^\prime$ that is provably $\epsilon$-close to $\pi$, 
and if the estimate was $\pi^\prime (  {\cal G}_{\rm BAD} ) \leq \epsilon_0$, 
we may readily infer that $\pi  (  {\cal G}_{\rm BAD} ) \leq \epsilon_0 \! + \! \epsilon$. 
On the other hand, if the estimate was $\pi^\prime (  {\cal G}_{\rm BAD} ) > \epsilon_0$, 
we may readily infer that $\pi  (  {\cal G}_{\rm BAD} ) > \epsilon_0 \! - \! \epsilon$. 
Thus, if we are able to fine-tune $\epsilon$, generation from $\pi^\prime$ becomes also
of fundamental predictive value. 

Our algorithms generate graphs from distributions $\epsilon$-close to those implied by 
$G(n,p)$, $G(n,\vec{w} ),$ $G(n , W ),$ and $G(n,{\cal P})$ respectively, 
at the cost of a multiplicative factor $O( \log \epsilon^{-1} )$ in the running time. 
This implies that we can set $\epsilon \! = \! n^{-c}$, for any constant $c$,
thus matching typical high probability statements for random graphs~\cite{b,d}. 
For historical reasons, we also mention that $\epsilon$-close sampling 
has been extensively used in theoretical computer science, 
for example in the context of approximate counting via Monte Carlo Markov chain simulation, 
among others. 

In particular, for $G(n,p)$, the running time
is $O\left( \mu \log n  \log \epsilon^{-1} (\log n+ \log p^{-1}) \right)$ in expectation where $\mu$ is the expected number of edges of the output graph and the factor $ (\log n+ \log p^{-1})$ accounts for representation and arithmetic.  We may obtain a high probability upper bound on the running time at a
cost of an additional $O(\log n )$ multiplicative factor. 

Finally, we should note that, for the sake of clarity in presentation,
we have not tried to optimize the poly-logarithmic factors at all points (especially in Section 5).
There are also points where the pseudocode might hint to ``difficult"
arithmetic (such as computing square roots). In such cases, 
we comment right below those pseudocodes, as to how these points can be bypassed. 

The rest of the paper is organized as follows. 
In Section 1.1 we quantify the claim that known algorithms for generating Erdos-Renyi $G(n,p)$ graphs 
with $m$ output edges run in $O(m)$ time. We note that these methods actually require $O(n)$ bits of arithmetic accuracy per step, which
clearly results in $O(mn)$ running times, when all computational aspects are considered.
Hence these methods for $G(n,p)$ (and hence all generalizations of $G(n,p)$, e.g. see ~\cite{mh}) are not efficient.
In Sections 2, 3, 4, and 5 we develop efficient algorithms to generate random graphs 
$\epsilon$-close to $G(n,p),$ $G(n,\vec{w}),$ $G(n,W),$ and $G(n,{\cal P})$ respectively.

\subsection{Previous work and Contributions}

The standard reference for generating a graph from $G(n,p)$ with running time $O( m )$, 
where $m$ is the number of output edges, is~\cite{bb}. 
This algorithm was recently extended to $G(n,\vec{w})$ in~\cite{mh}. 
For $G(n,p),$ the idea in~\cite{bb} is to order vertices and have each vertex decide 
its ``distance" or ``jump" to its next neighbor according to the ordering, 
where ``jumping" a vertex corresponds to an edge that is not present in the final output graph.
All vertices bypassed by such jumps will not be processed in the corresponding iteration.
Thus, the intuition is that the number of actual steps of the algorithm is proportional
to the number of edges present in the output graph, suggesting a $O(m)$ running time.

However, we argue below that the implementation of these ``jumps", 
which involve the simulation 
of a negative binomial with parameters $n$ and $p$,  
immediately introduce the necessity of $n$-bit accuracy per ``jump".
Therefore, the running time becomes $O( n m)$. 
The authors bypassed this problem by assuming constant running time for real number arithmetic and representation with arbitrary accuracy.
If  $n$-bit accuracy is compromised for any ${\rm poly} \log n$-bit accuracy, 
then the bias becomes immediately arbitrary.  

The algorithm in~\cite{bb} works as follows. First, the vertices are ordered, and the main observation is
that, at any given time during an iteration, 
the probability of generating the next edge after exactly $k$ trials 
is $(1\! - \! p)^{k-1}p$; i.e. waiting times for the edges are geometrically distributed. Let $q \! = \! 1\! - \! p$; to sample waiting times, 
each positive integer $k$ is assigned an interval $I_k \! \subseteq  \! [0,1)$ 
of length $q^{k-1}p.$ Realize that $\sum_{k=1}^{\infty}q^{k-1}p \! = \! 1,$ and if the intervals are contiguous starting at 0, 
then interval $I_k$ ends at $\sum_{i=1}^{k}q^{i-1}p \! = \! 1 \! - \! q^k$. Therefore, the waiting times can be sampled by randomly chosing $r \! \in \! [0,1)$ and selecting the smallest $k$ for which $r < 1 \! - \! q^k .$
The pseudocode from \cite{bb} is included bellow.
\begin{center}
\begin{fmpage}{8cm}

\hspace{0mm} $E \leftarrow \emptyset$; $v \leftarrow 1$; $w \leftarrow -1$;

\hspace{0mm} {\bf while} $v< n$ {\bf do}

\hspace{5mm} Draw $r  \! \in \! [0,1)$, uniformly at random;

\hspace{5mm} $w \leftarrow w\! + \! 1 \! + \! \lfloor     \frac {\log (1 - r)}{\log (1 - p)}                            \rfloor;$ 

\hspace{5mm} {\bf while} $w \geq v$ and $v \leq n$ {\bf do}

\hspace{10mm}  $w \leftarrow w \! - \! v$; $v \leftarrow v \! + \! 1$;

\hspace{5mm} {\bf if} $v < n$ {\bf then} $E \leftarrow E \cup \{ u,w \};$

\hspace{0mm} {\bf return} $(E);$
\end{fmpage}
\end{center}

It should be clear now that we need $r$ sampled with accuracy $O(n)$ bits in order to simulate fair sampling from all intervals.
If the arithmetic used
has accuracy $\alpha (n)$, then all intervals $I_k$ for $k > \alpha(n)$ will not be represented in the sampling. 
Thus, $\alpha (n) \! = \! o(n)$ introduces arbitrary (unquantified) sampling bias on every step of the algorithm. We also note that the arithmetic in~\cite{bb} involves the computation of discrete logarithms,
and this issue has been raised in~\cite{nlkb}, 
who however do not offer any solution with quantified performance. 
In~\cite{mh}, the same issues carried over for $G(n,\vec{w})$.

Furthermore, the authors in~\cite{bb} list an array of widely used software 
for random graph generation that implement inefficient algorithms. They note that such software 
provides running times tolerable for graphs up to tens of thousands of nodes. We remark that current technology
requires synthetic data involving much larger number of nodes (e.g. to simulate social networks.) 

Of course, the most natural way to sample $G(n,p)$ is to sample the number of edges $m$ from the 
binomial distribution $B( {n \choose 2 }, p)$, and then choose exactly $m$ out of the $n \choose 2$ edges at random. The main idea of our $G(n,p)$ generation algorithm is  a  Metropolis Markov chain to sample $\epsilon$-close from the binomial distribution. Our analysis accounts for all necessary bit-accuracy and arithmetic. Our running times are  comparable to known methods for exact binomial sampling (e.g., surveyed in \cite{ks}), however, only when the latter do not account for bit-accuracy and assume arbitrary bit arithmetic in constant time. Dropping these assumptions affects their running times and/or causes bias which has never been quantified. Our method to sample from the binomial distribution offers a rigorous quantification of the tradeoff between accuracy and running time, when all computational aspects are taken in to account. Therefore, it is of separate interest and may find other useful applications.

Our algorithms for $G(n,\vec{w}),$ $G(n,W),$ and $G(n,{\cal P})$ offer also efficient running times while taking into account all computational aspects. To the best of our knowledge, our work can be viewed as the first effort to 
simulate efficient generation of graphs from classical random graph models, 
while (a) taking into account fundamental implementational considerations
and (b) quantifying possible ``errors" in a way that can be useful in practice,
i.e. with quantified predictive value.

\section{Efficient Generation $\epsilon$-close to $G(n,p)$}
\label{gnp}

In this section $\pi$ is the probability distribution over ${\cal G}_n$ implied by $G(n,p)$. 
In particular, for any specific graph $G(V,E) \! \in \! {\cal G}_n$,
$\pi ( G(V,E) ) \! = \! p^{|E|} (1 \! - \! p )^{N - |E|} $, 
where $N \! = \! {n \choose 2}$.
It is obvious that $\pi ( |E| \! = \! k ) = {N \choose k} p^{k} (1 \! - \! p )^{N - k} $,
since there are ${N \choose k}$ distinct graphs in ${\cal G}_n$ with $k$ edges
and  $\pi$ assigns to all these graphs the same probability.
Thus, if  $B(N,p)$  is  the binomial distribution with parameters $N$ and $p$, 
then $\pi ( |E| \! = \! k ) = B(k;N,p)$.

Therefore, a natural two-step approach to generate a graph $G(V,E)$ according to $\pi$ is 
to sample $|E|$ from $B(N,p)$
and generate a random combination of $|E|$ out of $N$ possible edges. 
The second step can be implemented in time $O(\max\left\{ |E| , n \right\} \log n)$ 
(including representation and arithmetic)
using, for example,  the classic algorithm in~\cite{bf}. 
Sampling from $B(N,p)$, however, is much more involved, and it is analyzed in detail in Subsection~\ref{approxBNp}.

In Subsection~\ref{samplefast} we give algorithm 
{\bf Sample-G}$(n,p,\epsilon ) $ which uses the {\bf Coupling Markov Chain}~(\ref{coupleBNp}) 
to generate $|E|$, and include all implementational aspects (beyond mixing time)
that result in efficient sampling from a distribution $\epsilon$-close to $G(n,p)$.

\subsection{Markov Chains $\epsilon$-close to Binomial Distributions}
\label{approxBNp}

Let $B(N,p)$ be a binomial distribution with parameters $N \in \mathbb{N}$ and $0 \! < \! p\! < \! 1$. There are several methods to sample from the binomial distribution (for a detailed survey see, e.g.~\cite{ks}). However, we desire a rigorous quantification of the tradeoff between accuracy and running time. Current known techniques do not have this feature.

We design a Markov chain approach
to sample from a distribution that is $\epsilon$-close to $B(N,p)$,
in the sense of variation distance.
Let $\mu \! = \! Np$; the Markov chains are Metropolis-Hastings random walks on  a segment of the line 
around $\mu \pm  O(\sqrt{\mu \ln \epsilon^{-1}})$,
with expected coupling times (convergence rates)  $O( \mu \ln \epsilon^{-1})$.
The Markov chain is defined in~(\ref{mcBNp}) and the coupling which bounds 
convergence rate is defined in~(\ref{coupleBNp}). 
Lemma 4 and Theorem 5 establishes convergence and mixing time.

Let $0\! < \! \epsilon \! < \! 1$, $\xi \! = \! \mu \! - \! \lfloor \mu \rfloor,$ and  $\Delta  \geq \sqrt{ 4 N \max \left\{  p ,  \frac{4}{N} \ln \left( 2/\epsilon \right) \right\} \ln \left( 2/\epsilon \right) }.$
If $\xi \! \leq \! 1\! - \! p$, let $\bar{\mu} \! = \! \lfloor \mu \rfloor$; otherwise let $\bar{\mu} \! = \! \lceil \mu \rceil$, and define $\Delta^- \! = \! \min \{ \Delta , \bar{\mu} \}$ and $\Delta^+ \! = \!  \min \{ \Delta , N - \bar{\mu}\}$. Finally, let $B_{\Delta}(N,p)$ be the following probability distribution defined on the integer interval
$I \! = \! \left[ \bar{\mu}\! -\! \Delta^- ,  \bar{ \mu}\! + \!\Delta^+ \right]$,
\begin{equation}
\label{Bdelta}
B_{\Delta}(k; N,p) = \frac {B(k;N,p)}   
{ 1 - \sum\limits_{ x \not\in I}B(x;N,p) }
\end{equation}

\begin{fact}
\begin{equation}
\label{boundepsilon}
\sum_{ x \not\in I}B(x;N,p) < \epsilon
\end{equation}
\end{fact}

\begin{proof} For $p \! > \!  \frac{4}{N} \ln \left( 2/\epsilon \right)$ 
and $\delta \! = \! \sqrt{ \frac{4 \ln  \left( 2/\epsilon \right)}{\mu}} \! < \! 1$, Chernoff bounds suggest,
\[
\Pr \left[ \left|B(N,p) \! - \! \mu \right| > \sqrt{4  \mu \ln  \left( 2/\epsilon \right) }  
     \right]
= 
\Pr \left[ |B(N,p) \! - \! \mu | > \delta \mu
     \right]
< 2 e^{-\frac{\delta^2}{4}\mu}
= \epsilon
\]

Similarly, for $p \! \leq \! \frac{4}{N} \ln \left( 2/\epsilon \right)$ and $\delta \! = \! \frac{4 \ln  \left( 2/\epsilon \right)}{\mu},$ Chernoff bounds and $\mu \! \leq \! 4 \ln \left( 2/\epsilon \right)$ suggest, 
\[
\Pr \left[ |B(N,p) \! - \! \mu | > 4 \ln  \left( 2/\epsilon \right) \right]  = \Pr \left[ B(N,p) > (1+\delta)\mu \right] <  e^{-\frac{\delta^2}{2+\delta}\mu} < \epsilon.  
\]
Therefore, we have
$\Pr \left[ |B(N,p) \! - \! \mu | > \Delta \right] < \epsilon$ and  
$\Delta  \geq \sqrt{ 4 N \max \left\{  p ,  \frac{4}{N} \ln \left( 2/\epsilon \right) \right\} \ln \left( 2/\epsilon \right)  }$ for all $p$ and $0 \! < \! \epsilon \! < \! 1$.
Now (\ref{boundepsilon}) in Fact 1 follows immediately. \end{proof}

To sample from $B_{\Delta}(N,p)$ we define a {\bf Metropolis-Hastings Markov chain} $M$ on the interval $I$ with stationary distribution $B_{\Delta}(N,p).$ The transition probabilities are,
\begin{equation}
\label{mcBNp}
X_{t+1} = \left\{ \begin{array}{ll} 
                                  X_t       &   {\bf w.p.} ~~1/2  \\ 
                                  X_t+1  &   {\bf w.p.}~~\alpha^+(X_t)/4 \\
                                  X_t -1  &   {\bf w.p.}~~\alpha^-(X_t)/4  \\ 
		          X_t      &   {\bf w.p.}~~\frac{1-\alpha^+(X_t)}{4}+\frac{1- \alpha^-(X_t)}{4}\\ 
\end{array} \right.
\end{equation}
where the functions $\alpha^+(\cdot)$ and $\alpha^-(\cdot)$ are defined as usual for Metropolis-Hastings Markov chains,
\begin{displaymath}
\alpha^+ (k)  =  \left\{
                                         \begin{array}{ll}
                                              0     &   {\bf if}~k =   \bar{\mu}+ \Delta^+           \\
  \frac{N-k}{k+1}\frac{p}{1-p}   &   {\bf if}~ \bar{\mu} \leq k <  \bar{\mu}+ \Delta^+      \\
                                             1     &    {\bf  if}~ k < \bar{\mu} \\
                                         \end{array}
                           \right.
~~~~
\alpha^- (k)  =  \left\{
                                         \begin{array}{ll}
                                              1     &   {\bf if}~k > \bar{\mu}          \\
  \frac{k}{N-k+1}\frac{1-p}{p}   &   {\bf if}~   \bar{\mu} \geq k >  \bar{\mu}- \Delta^-   \\
                                              0    &    {\bf  if}~ k =   \bar{\mu}- \Delta^-  \\
                                         \end{array}
                           \right.
\end{displaymath}

\noindent
\begin{fact} The range of $\alpha^+(\cdot)$ and $\alpha^-(\cdot)$ is $[0,1]$. \end{fact} 

\begin{proof} This is ensured by the definition of $\bar{\mu}$ and can be verified by elementary calculations.\end{proof}

\noindent
\begin{fact} For any starting state (or probability distribution) $X_0 \in I$, $X_t$ converges to $B_{\Delta}(N,p)$. \end{fact}

\begin{proof} It is obvious that $X_t$ is ergodic. Convergence to $B_{\Delta}(N,p)$ follows by verifying detailed balance conditions. The details are in Appendix 1. \end{proof}

\noindent
To bound the mixing time of $M$, we define a {\bf coupling} $(X_t, Y_t)$ on $I \times I$ and analyze its coupling time. The transitions probabilities are,

\begin{equation}
\label{coupleBNp}
(X_{t+1} , Y_{t+1}) = \left\{
\begin{array}{ll} 
(X_t,Y_t+1)     & {\bf w.p.}~~\alpha^+(Y_t)/4\\
(X_t, Y_t-1)     & {\bf w.p.}~~\alpha^-(Y_t)/4 \\
(X_t , Y_t )      & {\bf w.p.}~~\frac{1- \alpha^+(Y_t)}{4}+\frac{1- \alpha^-(Y_t)}{4}\\ 
(X_t+1,Y_t)     & {\bf w.p.}~~\alpha^+(X_t)/4\\
(X_t-1, Y_t)  & {\bf w.p.}~~\alpha^-(X_t)/4 \\
(X_t , Y_t )      &  {\bf w.p.}~~\frac{1- \alpha^+(X_t)}{4}+\frac{1- \alpha^-(X_t)}{4}\\ 
 \end{array} 
                                   \right.
\end{equation}
while $ X_t \neq Y_t.$ Once $X_t\! = \!  Y_t,$ they remain equal for all future times following the transitions in~(\ref{mcBNp}).

\begin{lemma}
For the coupling $(X_t , Y_t) $ with $X_0 \! = \! \bar{\mu} \! + \! \Delta^+ $ and $Y_0 \! = \! \bar{\mu} \! - \! \Delta^-$,  let $T \! = \! \min_t \{ X_t \! = \! Y_t   \}$ be the coupling time. 
Then $X_t$ is distributed according to $B_{\Delta}(N,p)$, for all $t \geq T$.
\end{lemma}

\begin{proof}
Let $(\widehat{X}_t,\widehat{Y}_t)$ be the coupling $(X_t, Y_t)$ with 
$\widehat{X}_0 \! = \! \bar{\mu} \! + \! \Delta^+ $ and $\widehat{Y}_0$ 
sampled from the stationary distribution 
$B_{\Delta} (N,p)$. 
Thus $\widehat{Y}_t$ is distributed according to $B_{\Delta} (N,p)$, $\forall t$.
Notice that $X_0 \! = \! \hat{X}_0 \! \geq \! \hat{Y}_0 \! \geq \! Y_0$  implies
immediately  $X_t  \! = \! \widehat{X}_t \! \geq \! \widehat{Y}_t \! \geq \! Y_t$,  $\forall t$ by the monotonicity of the coupling. 
Thus, if $T \! = \! \min_{t} \{  X_t \! = \! Y_t \} $, then 
$X_T  \! = \! \widehat{X}_T \! = \widehat{Y}_T \! = \! Y_T$,
implying $X_T\! = \!\widehat{Y}_T.$ Therefore $X_T $ is 
distributed according to $B_{\Delta} (N,p),$ and $X_t$  is also distributed according to $B_{\Delta} (N,p)$, $\forall t \geq T$.
\end{proof}

\begin{thm}
$E\left[ T \right]$ is $O\left(\Delta^2\right)$ and, for any $c> 1$, $Pr \left[ T > 2 c \log n E\left[ T \right] \right] \leq n^{-c}$.
\end{thm}

\begin{proof} First note that 
$L = \Delta^+ \! + \! \Delta^- \! = \! (X_0 \! - \! Y_0 ) \! \geq \! (X_t \! - \! Y_t ) \! \geq \! 0$
and $T \! = \! \min_t \{ (X_t \! - \! Y_t ) \! = \! 0 \} $. 
Furthermore, the definition of the coupling in~(\ref{coupleBNp}) implies that $(X_{t+1} \! - \! Y_{t+1})$ is,
\begin{equation}
\label{timeBNp}
(X_{t+1} - Y_{t+1}) = \left\{
\begin{array}{ll} 
(X_t-Y_t) -1     & {\bf w.p.}~~\frac{\alpha^+(Y_t)}{4} + \frac{\alpha^-(X_t)}{4}    \\
(X_t-Y_t  ) +1   & {\bf w.p.}~~\frac{\alpha^-(Y_t)}{4} + \frac{\alpha^+(X_t)}{4} \\
(X_t - Y_t ) +0     & {\bf w.p.}~~1 - \frac{\alpha^+(Y_t)}{4} - \frac{\alpha^-(X_t)}{4} 
                                                 - \frac{\alpha^-(Y_t)}{4} - \frac{\alpha^+(X_t)}{4} \\ 
 \end{array} 
                                   \right.
\end{equation}

To bound ${\rm E}[T]$, we introduce a simpler process $Z_t$
which converges at least as fast at~(\ref{timeBNp}).
In particular, let $\{\alpha_t\}$ be any sequence with $0\! \leq\! \alpha_t \! \leq \!1$ for all $t \! \geq 0.$ Let $Z_0 \! = \! 0$, 
and let 
\begin{equation}
\label{processZ}
Z_{t+1} = \left\{
\begin{array}{ll}
                              Z_t+1 & {\bf w.p.}~~\frac{1}{4}\left( 1+ \alpha_t \right)  \\
                         \min\{ Z_t - 1 ,  0 \} & {\bf w.p.}~~\frac{1}{4}\left( 1+ \alpha_t \right)  \\
                             Z_t + 0 & {\bf w.p.}~~\frac{1}{2}\left( 1- \alpha_t \right)  \\
\end{array}
                 \right.
\end{equation}

\begin{lemma}
For some sequence  $\{\alpha_t\},$ $\label{boundZ} Z_{t} =  L - (X_{t} - Y_{t})$
\end{lemma}

\begin{proof} In order to proof this Lemma, we reduce the characterization of $(\ref{timeBNp})$ to three cases: $X_t > Y_t \geq \bar{\mu},$ $X_t > \bar{\mu} > Y_t$, and $\bar{\mu} \geq X_t > Y_t$. In each case, we show that $(\ref{timeBNp})$ is of the form $(\ref{processZ})$.  The details are in Appendix 2.
\end{proof}

\noindent
Let $\cal{A}$ be the set of all sequences in the interval $[0,1]$. Then,
\begin{equation}
\label{boundEZ}
{\rm E}[T]\! = \! {\rm E}[T \! = \! \min_t \{ (X_t \! - \! Y_t ) \! = \! 0 \}] ~ \leq \! 
\max_{
\begin{array}{c}
\alpha_{t} \in {\cal A}
\end{array}
}
{\rm E} [ \min_t \{ Z_t \! = \! L \}]
\end{equation}

\noindent
To bound the right-hand-side of~(\ref{boundEZ}), we use the following Lemma. 

\begin{lemma}
$\max\limits_{
\begin{array}{c}
\alpha_{t} \in {\cal A}
\end{array}
}
E [ \min_t \{ Z_t \! = \! L \}] \leq 2(L+1)^2$
\end{lemma}

\begin{proof}
The proof is a suitable adaptation of the proof of an equivalent statement for random walks on the integer line with reflecting barrier at zero. The details are in Appendix 3.
\end{proof}

Finally, using the bounds in (\ref{boundEZ}) and Lemma 7, we get the upper bound for the expectation of the coupling time $T$, 
\begin{displaymath}
{\rm E}\left[ T\right] \leq 
{\rm E} \left[ \min_t \{  Z_t \! = L\}\right]  \leq 2(L \! + \! 1 )^2 = O(\Delta^2)
\end{displaymath}

For the high probability statement of Theorem 5, Markov's inequality implies 
$\Pr \left[ T > 2 {\rm E}[T] \right] < 1/2$. If we view (pessimistically) the simulation 
of $2  c \log n {\rm E}\left[ T \right]$ steps 
of the process $(X_t, Y_t)$ as $c \log n$ independent experiments, 
each experiment consisting of running $2{\rm E}\left[  T \right]$ steps of the process $(X_t, Y_t)$, 
the probability that they all fail gives the bound 
$
\Pr \left[ T > 2  {\rm E}\left[ T \right] c \log n \right] < \left( \frac{1}{2} \right)^{c \log n } = n^{-c}$. This completes the proof of Theorem 5. \end{proof}

\subsection{Efficient Implementation for Sampling $\epsilon$-close to $G(n,p)$}
\label{samplefast}

We remark some important considerations for implementing {\bf Sample-G}$(n,p,\epsilon)$. The algorithm first uses standard multiplication algorithms to compute $\mu \!=\! N p$ in $O(\log^2 n \max\left\{ \log n , \log p^{-1} \right\} )$ time. To find 
$
\Delta  \geq \sqrt{ 4 N \max \left\{  p ,  \frac{4}{N} \ln \left( 2/\epsilon \right) \right\} \ln \left( 2/\epsilon \right)   }
$ efficiently, the algorithm may bound from above each term inside this expression 
by the corresponding smallest power of 2, thus making the computation of the resulting logarithm and square root elementary. This allows for a suitable $\Delta$ to be computed in total  
$O(\log^2 n \max \left\{ \log n , \log p^{-1}  , \log \epsilon^{-1} \right\} )$ time.

{\bf Sample-G}$(n,p,\epsilon)$  subsequently simulates the {\bf Coupling Markov Chain} (\ref{coupleBNp}) with 
starting state as in Theorem 5 to produce 
$k$ from $B_{\Delta} (N,p)$. Each step of the simulation involves the simulation of a step of 
the {\bf Metropolis-Hastings Markov Chain} (\ref{mcBNp}). By the definition of $\alpha^+(\cdot)$ 
and $\alpha^-(\cdot)$, each step of the simulation of (\ref{mcBNp}) 
can be completed in $O( \max \left\{ \log n , \log p^{-1}      \right\}    ) $, for a total of 
$O(\log n  \max \left\{ \log n , \log p^{-1}      \right\}    ) $ to update (write) $X_t$. 
The above, combined with Theorem 5, implies that $k$ can be sampled from $B_{\Delta} (N,p)$
in expected time $O(    \Delta^2     \log n  \max \left\{ \log n , \log p^{-1}      \right\}           )$.

Notice that $ \Delta^2 = O (     N \max \left\{ p , \frac{4}{N} \ln \left( 2/\epsilon \right) \right\}   \ln \left( 2/\epsilon \right)    ) =  O (    \max\left\{ \mu , 4 \ln \left( 2/\epsilon \right) \right\}  \ln  \left( 2/\epsilon \right)   )$.
We henceforth make the assumption that $\mu \geq 4 \ln \left( 2/\epsilon \right)$ 
(or else, a naive faster algorithm can be used instead), 
thus $ \Delta^2 = O \left(   \mu  \ln  \left( 2/\epsilon \right) \right)$. 
The total running time, is 
$O(   \mu    \log n   \ln \left( 2/\epsilon \right)  \max \left\{ \log n , \log p^{-1}      \right\}    )$, including all computations,
and the running time exceeds 
$O( c \mu \log^2 n           \ln  \left( 2/\epsilon \right) \max \left\{ \log n , \log p^{-1}      \right\}  )$ with probability $O(n^{-c})$ for any $c>1$.

Finally, {\bf Sample-G}$(n,p,\epsilon)$ chooses $k$ out of $N \! = \! {n \choose 2}$ edges and outputs 
$G(V,E).$ This step can be implemented in time $O(\max \left\{ k , n  \right\} \log N)$ 
(including representation and arithmetic)
using, for example,  the classic algorithm in~\cite{bf}. 

\begin{thm} Let $\pi$ be the distribution on ${\cal G}_n$ implied by $G(n,p)$. 
Algorithm {\bf Sample-G}$(n,p,\epsilon)$ outputs 
$G(V,E) \! \in \! {\cal G}_n$ sampled from a distribution $\pi^\prime$  on ${\cal G}_n$ 
that has total variation distance from $\pi$ at most $\epsilon$. Moreover,  for all $G(V,E) \! \in \! {\cal G}_n$, $\pi^\prime$
has the following additional properties:\\
$~~~~~\bullet~~\pi^\prime (G(V,E)) \geq \pi(G(V,E)) \implies \pi^\prime (G(V,E)) = \frac{\pi (G(V,E))}{1-\epsilon}$\\
$~~~~~\bullet~~\pi^\prime (G(V,E)) < \pi(G(V,E)) \implies \pi^\prime (G(V,E)) = 0$\\
(implementing the natural Coupling from the Past modification). Also, Algorithm {\bf Sample-G}$(n,p,\epsilon)$ runs in 
$O(      \mu    \log n   \ln  \left( 2/\epsilon \right) \max \left\{ \log n , \log p^{-1}      \right\}         )$ expected running time,
including all computations, and for any $c>1$, 
the probability that the running time exceeds 
$O( \mu c \log^2 n        \ln  \left( 2/\epsilon \right) \max \left\{ \log n , \log p^{-1}      \right\}    )$ 
is $O(n^{-c})$. 
\end{thm}

\begin{proof} Follows from Theorem 5, Fact 1, and the description of {\bf Sample-G}$(n,p,\epsilon)$
given above.\end{proof}

{\bf Remark.} There is a natural analogue to Theorem 8 for bipartite graphs. 
For integers $n_1$ and $n_2$ and $0 \! < \! p,\epsilon \! < 1$, 
let $G(n_1, n_2 , p )$ be the random bipartite graph with $n_1$ right vertices, $n_2$ left vertices, 
and edge probability between a right vertex and a left vertex $p$.
There is an algorithm {\bf Sample-G}$(n_1,n_2,p,\epsilon)$ which generates 
$G(V,E) \! \in \! {\cal G}_{n_1,n_2}$ in running time and with properties completely analogous
to those stated in Theorem 8. (In particular, all computations follow 
by replacing $n$ by $n_1+n_2$).

\section{Efficient Generation $\epsilon$-close to $G(n,\vec{w})$}
\label{gnw}

In $G(n,\vec{w})$, the input consists of $n$ real numbers $\vec{w} = w_1,w_2,...,w_n$ corresponding to a weight for each vertex. 
The probability of an edge $\{u,v\}$ is  given by $\min\{w_uw_v,1\},$ independently from all other edges.
Throughout this section, and by analogy to Section 2,
$\mu = {\rm E}[|E|]$ and $\pi$ is the distribution over all graphs on 
${\cal G}_n$ according to $G(n,\vec{w})$.
The main idea of the algorithm is to partition the vertices according to their weights, where the weights inside each partition class are within a multiplicative factor of 2.

Let $q = \max(|\log_2 \max_u\{w_u\}|,|\log_2 \min_u\{w_u\}|).$ In Phase 1, the algorithm rounds up each $w_u$ to the next power of 2, which partitions the vertices into $O(q)$ classes. Simultaneously, the ${n \choose 2}$ possible edges are partitioned according to the rounded weight of their endpoints. In Phases 2 and 3, the algorithm generates a random subgraph within each edge class independently. We observe that these random subgraphs are either in $G(n,p)$ or $G(n_1,n_2,p)$ for which we may use {\bf Sample-G}$(n,p,\epsilon')$ and {\bf Sample}-G$(n_1,n_2,p,\epsilon')$ of Section \ref{gnp}, for suitable choice of $\epsilon'$ that we shall determine. In Phase 4, the algorithm normalizes the output graph (usual accept-reject), so that each edge is sampled with probability $\min\{w_uw_v,1\}$ instead of the rounded weights.

\begin{center}
\begin{fmpage}{15cm}
{\bf Sample-G}$ (w_1,w_2,...,w_n,\epsilon)$

\hspace{0mm}  $k \leftarrow 0;$

\hspace{0mm} \textbf{\%Phase 1: Rounding}

\hspace{0mm} \textbf{for all} vertices $u$

\hspace{5mm} $w(u) \leftarrow 2^{\left\lceil\log_2(w_u)\right\rceil};$ $~~~~~i \leftarrow \log_2w(u);$

\hspace{5mm} {\bf if} $C_i = \emptyset$ {\bf then} $k \leftarrow k+1;$

\hspace{5mm} $C_i \leftarrow C_i \cup \{u\};$

\hspace{0mm} $\epsilon' \leftarrow 2\epsilon/(k(k+1));$

\hspace{0mm} \textbf{\%Phase 2: $G(n,p,\epsilon')$} 

\hspace{0mm} {\bf for all }$(C_i \ne \emptyset)$

\hspace{5mm} $E \leftarrow E~\cup$ {\bf Sample-G}$(|C_i|,\min\{{w(i)}^2,1\},\epsilon');$

\hspace{0mm} \textbf{\%Phase 3: $G(n_1,n_2,p,\epsilon')$} 

\hspace{0mm} {\bf for all }($C_i \ne \emptyset$ and $C_j \ne \emptyset$, $j > i$) 

\hspace{5mm} $E \leftarrow E~\cup$ {\bf Sample-G}$(|C_i|,|C_j|,\min\{w(i) w(j),1\},\epsilon');$ 

\hspace{0mm} \textbf{\%Phase 4: Normalization}

\hspace{0mm} {\bf foreach }($e=\{u,v\} \in E$) $E = E\setminus\{e\}~$ {\bf w.p.~} $1 - \frac{w_uw_v}{w(u)w(v)};$

\hspace{0mm} {\bf return} $(E)$;
\end{fmpage}
\end{center}

\begin{thm} 
{\bf Sample-G}$(w_1,...,w_n,\epsilon)$ generates a graph from a distribution $\pi'$ that is $\epsilon$-close to $\pi$ in expected time $O \left(\mu\log n\log (q/\epsilon) \left(\log n + q\right) + q^2\right)$. Moreover, for any constant $c$ the probability that the running time exceeds its expectation by a $2c\log n$ multiplicative factor is at most  $O(n^{-c})$.
\end{thm}

\begin{proof} The normalization happening in Phase 4 ensures that the rounding in Phase 1 has no net effect on the distribution the algorithm samples from. This may not be immediately obvious, since the edges of the graph constructed at the end of Phase 3 were not the result of fully independent sampling. However, it is tedious but straightforward 
to bound the probability of a graph $G(V,E)$ being the output at the end of Phase 3. That is,
$$
\prod_{\{ u,v \} \in E} w_uw_v \prod_{\{ u,v \}  \not\in E}  (1-w_uw_v)  ~ \leq ~   \Pr [G(V,E) ] 
~ \leq ~\frac{\prod_{\{ u,v \} \in E} w_uw_v \prod_{\{ u,v \}  \not\in E}  (1-w_uw_v)  }{1 - \epsilon}
$$

Phase 1 partitions the vertices into $k$ classes. Let $\pi_{ij}$ be the probability distribution according to $G(n,\vec{w})$ over subgraphs in edge class $[C_i,C_j]$. In Phases 2 and 3, Theorem 8 guarantees that, for each $i \leq j,$ the algorithm samples from a distribution $\pi_{i,j}'$ which is $\epsilon$-close to $\pi_{ij}.$

Let $X$ be the set of all graphs from which $\pi'$ does not sample, and for each $i$ and $j$ let $X_{ij}$ be the set of subgraphs from which $\pi'_{ij}$ does not sample. To sample from $\pi',$ the algorithm samples once from each $\pi'_{ij}.$ Therefore, $G \in X$ if and only if there exists a subgraph $H$ of $G$ such that $H \in X_{ij}$ for some $i$ and $j.$  Using union bound,
\begin{eqnarray}
\pi(X) \leq \sum\limits_{i=1}^k\sum\limits_{j=i}^k \pi_{i,j}(X_{ij}) \leq \frac{k(k+1)\epsilon'}{2} = \epsilon \nonumber
\end{eqnarray}

Let $G$ be any of the possible output graphs of the algorithm, and let $H_{ij}$ be the subgraph of $G$ induced by the vertices in classes $C_i$ and $C_j.$ By Theorem 8,
\begin{eqnarray}
\pi(G) &=& \prod\limits_{i=1}^k\prod\limits_{j=i}^k~\pi_{ij}(H_{ij}) = \prod\limits_{i=1}^k\prod\limits_{j=i}^k~\pi'_{ij}(H_{ij})(1-\pi_{ij}(X_{ij}))  \leq \prod\limits_{i=1}^k\prod\limits_{j=i}^k~\pi'_{ij}(H_{ij}) = \pi'(G) \nonumber
\end{eqnarray}

Hence, $\max_{S \subseteq {\cal G}_n} |\pi(S)-\pi'(S)|$ = $\max_{S \subseteq {\cal G}_n} |\pi(S \cap X) + \pi(S \setminus X) - \pi'(S \cap X) - \pi'(S \setminus X) | \leq \epsilon.$
 
To analyze the running time, let $T_1, T_2, T_3$ and $T_4$ be the running times for each of the four phases in the algorithm. Then, $T_1 = O(qn)$. Let $w(i)$ be the rounded up weight and $|C_i| = n_i$ for each $i$. By Theorem 8,
\begin{eqnarray}
T_2 &=& O \left(\sum\limits_{i=1}^k{n_i \choose 2} \min\{w(i)^2,1\}\cdot \log n \cdot \log (\epsilon')^{-1}\left(\log n + q \right)  \right) \nonumber\\
T_3 &=& O \left( \sum\limits_{i=1}^k \sum\limits_{j=i+1}^k n_in_j \min\{w(i)w(j),1\} \cdot \log n \cdot \log (\epsilon')^{-1}\cdot\left(\log n +  q \right) \right) \nonumber
\end{eqnarray}

Thus, $T_2 + T_3 = O \left(\mu' \log n \log (\epsilon')^{-1} \left(\log n + q \right) + k^2\right)$ where $\mu'$ is the expected number of edges of the output graph using the rounded weights. Given that $T_4 = O(q \mu'),$ $\mu \leq \mu' \leq 4\mu,$ and $k=O(q),$ the total expected running time of the algorithm is $O \left(\mu \log n \log (q/\epsilon) \left(\log n + q\right)  + q^2 \right).$
 
Finally, we observe that Markov's inequality implies $\Pr \left[ T > 2 {\rm E}[ T ] \right ] \leq 1/2$. Considering the worst case where each group of $c\log n$ steps is an independent experiment, we get $\Pr \left[ T > 2c\log n {\rm E}[ T ] \right ]  = O( n^{-c})$. This completes the proof of Theorem 9. \end{proof}

\section{Efficient Generation $\epsilon$-close to $G(n,W)$}
\label{gnW}

In $G(n,W),$ the input consists of an $n \times d$ matrix $W$ containing $n$ vectors: $W = (\vec{w_1},...,\vec{w_n})$, where each vector $\vec{w_u} \in \mathbb{R}^d$ corresponds to a vertex $u.$ The probability of each edge $\{u,v\}$ is given by $\min\{\langle \vec{w_u},\vec{w_v} \rangle,1\},$ independently from all other edges. Throughout this section, and by analogy to previous sections, let $\mu = {\rm E}[|E|],$ $\pi$ be the distribution over ${\cal G}_n$ according to $G(n,W),$ and $q = \max(|\log \max_u\{||\vec{w_u}||\}|,|\log  \min_u\{\vec{||w_u||}\}|)$.

\begin{center}
\begin{fmpage}{9cm}
{\bf Sample-G}$ (W,\epsilon )$

\hspace{0mm} \textbf{\%Phase 1: Rounding}

\hspace{0mm} \textbf{for} $u = 1$ {\bf to} $n$ 

\hspace{5mm} $L(u) \leftarrow \sqrt{\sum_{k=1}^d w_{uk}^2};$

\hspace{0mm} \textbf{\%Phase 2: $G(n,\vec{L},\epsilon)$} 

\hspace{0mm} $E \leftarrow$ {\bf Sample-G}$(L,\epsilon);$

\hspace{0mm} \textbf{\%Phase 3: Normalization}

\hspace{0mm} {\bf foreach }($e=\{u,v\} \in E$) 

\hspace{5mm} $E = E\setminus\{e\}~$ {\bf w.p.~} $1 - \frac{L(u)L(v)}{\left\langle \vec{w_u},\vec{w_v} \right\rangle};$

\hspace{0mm}  {\bf return} $(E);$
\end{fmpage}
\end{center}

The main idea of the algorithm is to reduce sampling from $G(n,W)$ to repeated sampling from $G(n,\vec{w})$ using properties of the inner product to round and normalize. The reduction uses the fact that  $\langle \vec{w_u},\vec{w_v} \rangle = ||\vec{w_u}||||\vec{w_v}||cos(\varphi(u,v))$ where $\varphi(u,v)$ is the angle between vectors $\vec{w_u}$ and $\vec{w_v}.$  In Phase 1, {\bf Sample-G}$ (W,\epsilon )$ assigns to each vertex $u$ a weight equal to its length $l_u = ||\vec{w_u}||.$ In Phase 2, the algorithm calls {\bf Sample-G}$(l_1,...,l_n,\epsilon')$ with a suitable choice of $\epsilon'$ we shall determine. Notice that at this point, each edge $\{u,v\}$ has been sampled with probability $||\vec{w_u}||||\vec{w_v}||$ instead of $\left\langle \vec{w_u},\vec{w_v} \right\rangle.$ In Phase 3, the algorithm normalizes the output graph (usual accept-reject). 

\begin{thm}
{\bf Sample-G}$(W,\epsilon)$ generates a graph from a distribution $\pi'$ that is $\epsilon$-close to $\pi$ in expected time $O \left(d\mu \log n \log (q/\epsilon) (\log n + q) + q^2 \right).$ Moreover, the probability that the running time exceeds its expectation by a $2c\log n$ multiplicative factor is at most  $O(n^{-c})$.
\end{thm}

\begin{proof} The normalization in Phase 3 ensures that the rounding in Phase 1 has no net effect on the distribution the algorithm samples from. Thus, by Theorem 9, $\pi'$ is $\epsilon$-close to $\pi.$ In Phase 1, one could round up each $L(u)$ to an even power of two simplifying the calculation of the square root with no effect in the running time (this step was left out of the pseudocode for clarity). Therefore, Phase 1 takes $O(qdn)$ time. If $\mu'$ is the expected number of edges of the output graph prior to Phase 3, then, Theorem 9 guarantees that {\bf Sample-G}$ (W,\epsilon )$ runs in $O \left(\mu' \log n \log (q/\epsilon) (\log n + q) + q^2 \right)$ time on expectation. To complete the proof, we use the following Lemma,

\begin{lemma} $d\mu \geq \mu'$ \end{lemma}

\begin{proof} For all $u$ and $v$ for which $\langle \vec{w_u},\vec{w_v} \rangle > 1,$ the generated subgraph graph will be complete and can be obtained trivially in $O(qn\log n).$ Hence, without loss of generality, we may assume that $\langle \vec{w_u},\vec{w_v} \rangle \leq 1$ for all $u$ and $v.$

Oberve that $\mu \!=\! \sum\limits_{u=1}^n \sum\limits_{v=u+1}^n\langle \vec{w_u},\vec{w_v} \rangle$ and $\mu' \!=\! \sum\limits_{u=1}^n \sum\limits_{v=u+1}^n ||\vec{w_u}||||\vec{w_v}||.$ When $u\!=\!v,$ $\langle \vec{w_u},\vec{w_v} \rangle \!=\! ||\vec{w_u}||||\vec{w_v}||,$ thus, we can show instead,
\begin{eqnarray}
\sum\limits_{u=1}^n \sum\limits_{v=1}^n \left\langle \vec{w_u},\vec{w_v} \right\rangle \geq  \frac{1}{d} \sum\limits_{u=1}^n \sum\limits_{v=1}^n ||\vec{w_u}||||\vec{w_v}|| \nonumber
\end{eqnarray}

Let $w_{uk}$ denote the components of $\vec{w_u}$ for each $u.$ Then $w_{uk} = l_u\cos (\varphi_{uk})$ where $\varphi_{uk}$ is the angle between $\vec{w_u}$ and the $k$-th dimension's axis. From the definition of inner product follows that $\cos(\varphi(u,v)) =  \sum\limits_{k=1}^d \cos(\varphi_{uk}) \cdot \cos(\varphi_{vk}).$ Therefore,\begin{eqnarray}
\sum\limits_{u=1}^n \sum\limits_{v=1}^n \left\langle \vec{w_u},\vec{w_v} \right\rangle &=& \sum\limits_{u=1}^n \sum\limits_{v=1}^n l_ul_v \left\lgroup \sum\limits_{k=1}^d \cos(\varphi_{ik}) \cdot \cos(\varphi_{jk})  \right\rgroup  \nonumber \\
&=& \sum\limits_{k=1}^d \left\lgroup \sum\limits_{u=1}^n l_u \cdot \cos(\varphi_{uk}) \left\lgroup \sum\limits_{v=1}^n l_v \cdot \cos(\varphi_{vk}) \right\rgroup \right\rgroup = \sum\limits_{k=1}^d \left\lgroup \sum\limits_{u=1}^n l_u \cdot \cos(\varphi_{uk}) \right\rgroup^2 \nonumber 
\end{eqnarray}
Using Jensen's inequality for the first bound, and repeated triangular inequality for the second,
\begin{eqnarray}
\mu = \sum\limits_{k=1}^d \left\lgroup \sum\limits_{u=1}^n l_u \cdot \cos(\varphi_{uk}) \right\rgroup^2 \geq \frac{\left\lgroup \sum\limits_{k=1}^d  \sum\limits_{u=1}^n l_u \cdot \cos(\varphi_{uk}) \right\rgroup^2}{d} 
\geq \frac{\left\lgroup \sum\limits_{u=1}^n l_u \right\rgroup^2}{d} = \frac{\sum\limits_{u=1}^n \sum\limits_{v=1}^n l_ul_v}{d} = \frac{\mu'}{d} \nonumber
\end{eqnarray} \end{proof}

\noindent 
Thus, {\bf Sample-G}$(W,\epsilon)$ runs in $O \left(d \mu \log n \log (q/\epsilon) (\log n + q) + q^2 \right)$ expected time. The high probability statement also follows immediately from Theorem 9.\end{proof}

\section{Efficient Generation $\epsilon$-close to $G(n,\cal{P})$}
\label{K}

In $G(n,\cal{P})$, the input consists of a positive integer $k$ and a $d \times d$ ``initiator" matrix ${\cal P} \!=\! (\theta_{ij}).$ Define $K_k$ recursively using the Kronecker product of matrices \cite{leckfz},
\begin{center}
$K_k = {\cal P} \otimes K_{k-1} = \begin{pmatrix} \theta_{11}K_{k-1} & \theta_{12}K_{k-1}& \ldots & \theta_{1n}K_{k-1} \\ \vdots & \vdots & \ddots & \vdots \\ \theta_{n1}K_{k-1} & \theta_{n2}K_{k-1} & \ldots & \theta_{nn}K_{k-1} \end{pmatrix}$
\end{center}
where $K_0$ is the identity. In $G(n,\cal{P}),$ the probability of an edge $\{u,v\}$ is given by $\min \{K_k(u,v),1\},$ independently from all other edges. Throughout this section, $n = d^k$ is the number of vertices of the output graph, $\mu = {\rm E}[|E|],$ $\pi$ is the distribution over all graphs on ${\cal G}_n$ according to $G(n,{\cal P}),$ and $q = \max(|\log_2 \max_{ij}\{\theta_{ij}\}|,|\log_2 \min_{ij}\{\theta_{ij}\}|).$
 
The main idea of the algorithm is to reduce sampling from $G(n,\cal{P})$ to repeated sampling from the binomial distribution using the method in Section 2. {\bf Sample-G$(n,{\cal P},\epsilon) $} partitions the edges according to their probability of occurrence. The probability of occurrence for each edge $e$ is given by $p_e = \theta_1\theta_2...\theta_k$ where each $\theta_i$ is an entry of ${\cal P}.$ Hence, $p_e = \theta_{11}^{\alpha_{11}}\theta_{12}^{\alpha_{12}}...\theta_{dd}^{\alpha_{dd}}$ for a suitable sequence of $\alpha_{ij}$'s. The number of edges with probability of occurrence $p_e$ is $N_e = \frac{k!}{\alpha_{11}!...\alpha_{dd}!}.$ Therefore, there is a natural two-step approach to generate the random subgraph from each edge class. First sample $|E|$ from $B(N_e,p_e),$ using the method described in Section 2, and then generate a random combination of $|E|$ out of $N_e$ possible edges, using for example the classic method in .

\begin{center}
\begin{fmpage}{13cm}
{\bf Sample-G$(n,{\cal P},\epsilon) $}

\hspace{0mm} $\epsilon' \leftarrow \epsilon/k^{d^2};$

\hspace{0mm} {\bf foreach} ($(\alpha_{11},\alpha_{12},...,\alpha_{dd}) \in \mathbb{Z}^{d^2}:$ $\sum_{i=1}^d\sum_{j=1}^d \alpha_{ij} \! = \! k$ and $\alpha_{ij} \! \geq \!0$ $\forall i,j$)

\hspace{5mm} $N \leftarrow \frac{k!}{\alpha_{11}!\alpha_{12}!...\alpha_{dd}!}$;

\hspace{5mm} $p \leftarrow \prod\limits_{i=1}^d \prod\limits_{j=i}^d \theta_{ij}^{\alpha_{ij}};$

\hspace{5mm} $M \leftarrow$ {\bf Sample-B}$(N,p,\epsilon');$

\hspace{5mm} $\{c_1,...,c_M\} \leftarrow$ {\bf Random-Combination}$(N,M);$

\hspace{5mm} {\bf for} $i=1$ {\bf to} $M~~~E \leftarrow E \cup \{e_{c_i}\};$

\hspace{0mm} {\bf return} $(E);$
\end{fmpage}
\end{center}

The pseudocode may hint to ineficient arithmetic operations, but it is not the case. Establishing two suitable orders, one among edge classes and one among edges within each class,
we can find the next edge class easily in $O(d^2)$ time. Similarly, searching for the $c_i$-th edge in an edge class can be done in $O(d^4).$ Therefore, finding the edges corresponding to the random combination $\{c_1,...,c_M\}$ takes $O(d^4M)$ time. 

$N$ is initialized to $k!,$ and it is recomputed in each step with only one single multiplication operation. Computing $p$ involves $d^2$ multiplications per step, and using standard multiplication algorithms this can be done in $O(d^2k^2\log^2q)$. Therefore, we have a very crude bound of $O(d^2k^2\log^2 q+d^4M)$ for the cost of arithmetic and auxiliary operations per step, where, in practice $d$ is a very small constant \cite{leckfz}, we have not tried to optimize the exponent of $d$ for the sake of clarity.

\begin{thm}
{\bf Sample-G}$(n,{\cal P},\epsilon)$  generates a graph from a distribution $\pi'$ that is $\epsilon$-close to $\pi$ in expected time $O(d^4 \mu\log(\epsilon^{-1})\max\{\log n,\log q\} + d^2 (\log n)^{d^2+2} \log^2 q).$ Moreover, the probability that the running time exceeds its expectation by a $2c\log n$ multiplicative factor is at most  $O(n^{-c})$.
\end{thm}

\begin{proof} The proof is very similar to the proof of Theorem 9. The details are in Appendix 4.\end{proof}

\pagebreak

\pagebreak

\section*{Appendix 1: Proof of Fact 3}

{\bf Fact 3.}
{\em For any starting state (or probability distribution) $X_0 \in I$, $X_t$ converges to $B_{\Delta}(N,p)$.}

\begin{proof} It is obvious that $X_t$ is ergodic.
Convergence to $B_{\Delta}(N,p)$ follows by verifying detailed balance conditions. For $k$ in the range $\bar{\mu}\!+\!1 \leq k \leq \bar{\mu}\!+\!\Delta^+$, detailed balance conditions are
equivalent to:
\begin{eqnarray*}
B(k\! - \! 1 ; N,p) \Pr [X_{t+1} \! = \! k | X_t \! = \! k \! - \! 1 ] &  =  &  
B(k\! - \! 1 ; N,p) ~\alpha^+ (k-1) \\
      & = & \frac{N!}{(k-1)! (N-k+1)!}p^{k-1}(1\! - \! p)^{N-k+1} \frac {N-k+1}{k}\frac{p}{1-p}  \\
      & = & \frac{N!}{(k)! (N-k)!}p^{k}(1\! - \! p)^{N-k} \\
      & = & B(k;N,p) \\ 
     & =  & B(k;N,p) ~ \alpha^- (k) \\
           & = &  B(k;N,p) \Pr [X_{t+1} \! = \! k \! - \! 1 | X_t \! = \! k]
\end{eqnarray*}
For  $k$ in the range $\bar{\mu}\!-\!\Delta^- \! < k \leq \bar{\mu}$, detailed balance conditions are equivalent to:
\begin{eqnarray*}
B(k; N,p) \Pr [X_{t+1} \! = \! k\! - \! 1 | X_t \! = k ] &  =  &  
B(k; N,p) ~\alpha^- (k) \\
      & = & \frac{N!}{(k)! (N-k)!}p^{k}(1\! - \! p)^{N-k} \frac{k}{N-k+1}\frac{1-p}{p} \\
 & = & \frac{N!}{(k-1)! (N-k+1)!}p^{k-1}(1\! - \! p)^{N-k+1} \\
       & = & B(\! -\! 1;N,p) \\ 
     & =  & B(k\! - \! 1;N,p)~\alpha^+(k-1) \\
           & = &  B(k\! - \! 1 ;N,p) \Pr [X_{t+1} \! =\!  k | X_t  \! = \! k \! - \! 1]
\end{eqnarray*}
This completes the proof of Fact 3. \end{proof}

\section*{Appendix 2: Proof of Lemma 6}

{\bf Lemma 6.}{\rm For some sequence  $\{\alpha_t\},$ $\label{boundZ} Z_{t} =  L - (X_{t+1} - Y_{t+1}).$}

\begin{proof} For this, we reduce the characterization of $(\ref{timeBNp})$ to Cases 1, 2 and 3 below. 
In each case, ($(\ref{bound1})$, $(\ref{bound2})$ and $(\ref{bound3})$), we show 
that $(\ref{timeBNp})$ is of the form $(\ref{processZ})$.

\noindent
{\bf Case 1:} 
If $X_t > Y_t \geq \bar{\mu}$, then $\alpha^-(X_t) \! = \! 1$, 
and letting $Y_t\! = \! k$, $X_t \! = \! k\! + \! x$, $x \! \geq \! 1$, we have 
\begin{displaymath}
(X_{t+1} - Y_{t+1}) = \left\{
\begin{array}{ll} 
(X_t-Y_t) -1     & {\bf w.p.}~~\frac{1}{4}\left( \frac{N-k}{k+1} \right) 
                                                 \frac{p}{1-p} + \frac{1}{4}    \\
(X_t-Y_t  ) +1   & {\bf w.p.}~~\frac{1}{4}\left( \frac{N-k-x}{k+x+1}\right) 
                                                    \frac{p}{1-p} + \frac{\alpha^-(Y_t)}{4}  \\
(X_t - Y_t ) +0     & {\bf w.p.}~~\frac{1}{2} - \frac{1}{4} \left( \frac{N-k}{k+1} +\frac{N-k-x}{k+x+1}  \right) 
                                               \frac{p}{1-p} + \frac{1}{4} -\frac{ \alpha^-(Y_t)}{4}\\ 
 \end{array} 
                                   \right.
\end{displaymath}
Realize that $  \frac{N-k}{k+1} > \frac{N-k-x}{k+x+1}  $ 
and $\alpha^-(Y_t) \leq 1$. 
Thus moving probability $\frac{1}{4} \left( \frac{N-k}{k+1} - \frac{N-k-x}{k+x+1}  \right)  \frac{p}{1-p} $
from the $-1$ level to the $ 0 $ level, 
and moving probability $\frac{1}{4} -\frac{ \alpha^-(Y_t)}{4} $
from the $0$ level to the $+1$ level, we may bound 
\begin{equation}
\label{bound1}
(X_{t+1} - Y_{t+1}) \leq \left\{
\begin{array}{ll} 
(X_t-Y_t) -1     & {\bf w.p.}~~\frac{1}{4}
                                                      \left( 1 +
                                                                           \left( \frac{N-k-x}{k+x+1} \right)  \frac{p}{1-p}  
                                                     \right)    \\
(X_t-Y_t  ) +1   & {\bf w.p.}~~\frac{1}{4}
                                                     \left( 1 + 
                                                                            \left( \frac{N-k-x}{k+x+1}\right)    \frac{p}{1-p} 
                                                     \right)  \\
(X_t - Y_t ) +0     & {\bf w.p.}~~\frac{1}{2}\left( 1 -\frac{N-k-x}{k+x+1}\frac{p}{1-p} \right)\\
                                                \end{array} 
                                   \right.
\end{equation}

\noindent
{\bf Case 2:} 
If $X_t > \bar{\mu} > Y_t$, then $\alpha^+(Y_t) \! = \! \alpha^-(X_t) \! = \! 1$, 
and we have 
\begin{displaymath}
(X_{t+1} - Y_{t+1}) = \left\{
\begin{array}{ll} 
(X_t-Y_t) -1     & {\bf w.p.}~~\frac{1}{2} \\
(X_t-Y_t  ) +1   & {\bf w.p.}~~\frac{\alpha^-(Y_t)}{4} + \frac{\alpha^+(X_t)}{4} \\
(X_t - Y_t )      & {\bf w.p.}~~\frac{1}{2} - 
                                                  \frac{\alpha^-(Y_t)}{4} - \frac{\alpha^+(X_t)}{4} \\ 
 \end{array} 
                                   \right.
\end{displaymath}
Moving $1/4$ probability from the $-1$ level to the $0$ level, 
and $\frac{1}{4}- \frac{\alpha^-(Y_t)}{4} - \frac{\alpha^+(X_t)}{4}$ probability from the 
$0$ level to the $+1$ level, we may bound
\begin{equation}
\label{bound2}
(X_{t+1} - Y_{t+1}) \leq \left\{
\begin{array}{ll} 
(X_t-Y_t) -1     & {\bf w.p.}~~\frac{1}{4} \\
(X_t-Y_t  ) +1   & {\bf w.p.}~~\frac{1}{4} \\
(X_t - Y_t )      & {\bf w.p.}~~\frac{1}{2} \\
\end{array} 
                                   \right.
\end{equation}

\noindent
{\bf Case 3:}
If $\bar{\mu} \geq X_t > Y_t$, then $\alpha^+(Y_t) \! = \! 1$, 
and letting $Y_t\! = \! k$, $X_t \! = \! k\! + \! x$, $x \! \geq \! 1$, we have 
\begin{displaymath}
(X_{t+1} - Y_{t+1}) = \left\{
\begin{array}{ll} 
(X_t-Y_t) -1     & {\bf w.p.}~~\frac{1}{4}
                                          \left( \frac{k+x}{N-k-x+1} \right)    \frac{1-p}{p} + \frac{1}{4}    \\
(X_t-Y_t  ) +1   & {\bf w.p.}~~\frac{1}{4}\left( \frac{k}{N-k+1}\right) 
                                                    \frac{1-p}{p} + \frac{\alpha^+(X_t)}{4}  \\
(X_t - Y_t )  +0    & {\bf w.p.}~~\frac{1}{2} - \frac{1}{4} \left( \frac{k+x}{N-k-x+1} +  \frac{k}{N-k+1}  \right) 
                                               \frac{1-p}{p} + \frac{1}{4} - \frac{\alpha^+(X_t)}{4}\\ 
 \end{array} 
                                   \right.
\end{displaymath}
Realize that $ \frac{k+x}{N-k-x+1} >  \frac{k}{N-k+1} $ and $\alpha^+(X_t) \leq 1$. 
Thus moving probability $\frac{1}{4} \left(   \frac{k+x}{N-k-x+1} -  \frac{k}{N-k+1} \right)  \frac{1-p}{p} $
from the $-1$ level to the $ 0 $ level, 
and moving probability $\frac{1}{4} -\frac{ \alpha^+(X_t)}{4} $
from the $0$ level to the $+1$ level, we may bound 
\begin{equation}
\label{bound3}
(X_{t+1} - Y_{t+1}) \leq \left\{
\begin{array}{ll} 
(X_t-Y_t) -1     & {\bf w.p.}~~\frac{1}{4}  \left( 1 + \left( \frac{k}{N-k+1} \right) 
                                                 \frac{1-p}{p}   \right) \\
(X_t-Y_t  ) +1   & {\bf w.p.}~~\frac{1}{4}    \left( 1 + \left(\frac{k}{N-k+1} \right) 
                                                    \frac{1-p}{p} \right) \\
(X_t - Y_t ) +0     & {\bf w.p.}~~\frac{1}{2}\left( 1 -\frac{k}{N-k+1}\frac{1-p}{p} \right)\\
                                                \end{array} 
                                   \right.
\end{equation}\end{proof}

\section*{Appendix 3: Proof of Lemma 7}

{\bf Lemma 7.}{\rm 
$\max\limits_{
\begin{array}{c}
\alpha_{t} \in {\cal A}
\end{array}
}
E [ \min_t \{ Z_t \! = \! L \}] \leq 2(L+1)^2$}

\begin{proof}
The proof is a suitable adaptation of the proof for random walks on the integer line with reflecting barrier at zero. For $k \geq 0$, we argue inductively that
\begin{equation}
\label{boundstepZ}
f(k+1) = 
\max_{
\begin{array}{c}
\alpha_{\tau} \in {\cal A}
\end{array}
}
\max_{t^\prime} ~
E \left[ \min_t \{  Z_{t^\prime+t} = k+1 | Z_{t^{\prime}} = k  \}    \right ] \leq 4(k+1)
\end{equation}
The base case $f(1)\!\leq\! 4$ is obvious. 
For the inductive step, the recursive definition in (\ref{processZ}) and the definition of $f(k+1)$ in (\ref{boundstepZ}) imply
that for some $0 \! \leq \! \alpha \! \leq \! 1$
$$
f(k+1)    \leq   \frac{1}{4}\left( 1 + \alpha \right)
              + \frac{1}{2}\left(   1 - \alpha  \right) \left( 1 +   f(k+1)  \right)  
              + \frac{1}{4}\left( 1 + \alpha \right) \left( 1 +  f(k) + f(k+1) \right)
$$
or equivalently, 
\begin{eqnarray*}
\frac{1}{4}\left( 1 + \alpha \right) f(k+1)   &\leq&  1 + \frac{1}{4}\left( 1 + \alpha \right)  f(k)\\
f(k+1) & \leq & \frac{4}{1 + \alpha } + f(k) \\
           & \leq & 4 + f(k)  \\
           & \leq & 4 + 4k ~~~~~~~~~~~~~~~~~~~~~\mbox{(by the inductive hypothesis)}\\
           & = & 4(k+1)
\end{eqnarray*}
thus establishing \ref{boundstepZ}. Combining (\ref{boundEZ}) and (\ref{boundstepZ}) we get 

\begin{eqnarray}
\max\limits_{
\begin{array}{c}
\alpha_{\tau} \in {\cal A}
\end{array}
}
E [ \min_t \{ Z_t \! = \! (\Delta^+ \! + \! \Delta^- )  ]  
&\leq& 
 \sum\limits_{k=1}^{\Delta^+ \! + \! \Delta^-} 
\max\limits_{
\begin{array}{c}
\alpha_{\tau} \in {\cal A}
\end{array}
}
\max\limits_{t^\prime} ~
E \left[ \min_t \{  Z_{t^\prime+t} = k+1 | Z_{t^{\prime}} = k  \}    \right ] \\
&\leq& \sum_{k=1}^{\Delta^+ \! + \! \Delta^-} f(k) 
\leq 4 \frac{(\Delta^+ \! + \! \Delta^-) (\Delta^+ \! + \! \Delta^- +1)}{2} \\
\label{boundEZfinal}
&\leq& 2(\Delta^+ \! + \! \Delta^- \! + \! 1 )^2
\end{eqnarray} \end{proof}

\section*{Appendix 4: Proof of Theorem 12}

{\bf Theorem 12.} {\em
{\bf Sample-G}$(n,{\cal P},\epsilon)$  generates a graph from a distribution $\pi'$ that is $\epsilon$-close to $\pi$ in expected time $O\left(d^4 \mu\log(\epsilon^{-1})\max\{\log n,\log q\} + d^2 (\log n)^{d^2+2} \log^2 q\right)$ Moreover, the probability that the running time exceeds its expectation by a $2c\log n$ multiplicative factor is at most  $O(n^{-c})$.}

\begin{proof} Taking $\epsilon' = \epsilon/k^{r},$ we may follow the exact same steps as in the proof of Theorem 9 to show that $\pi'$ is $\epsilon$-close to $\pi.$ In each step, {\bf Sample-B}$(N,p)$ and {\bf Random-Combination}$(N,M)$ are called for a total expected running time of $O\left( \mu\log(k^{d^2}/\epsilon')\max\{\log n,\log q\}\right),$ using Theorem 5 and the method to sample random combinations discussed in \cite{bf}. The cost of arithmetic and auxiliary operations per edge class is $O(d^4k^2\log^2 q),$ and there are $k^{d^2}$ classes with $k=O(\log n)$. Therefore, the overall running time of the algorithm is $O\left(d^4 \mu\log(\epsilon'^-1)\max\{\log n,\log \theta\} + d^4 (\log n)^{d^2+2} \log^2 q\right)$. Notice that some $\log\log n$ factors are omitted, and, as mentioned before, $d$ is a very small constant. The proof of the high probability statement for the running time goes exactly as in Theorem 9. \end{proof}

\end{document}